\documentclass[
  10pt,
  nofootinbib,
  superscriptaddress,
  onecolumn,
  amsfonts,
  amssymb,
  amsmath,
  aps,
  pra,
  floatfix
]{revtex4-2}
\usepackage{mathtools}
\usepackage{mathrsfs}
\usepackage{amsthm}
\usepackage{bm}
\usepackage{array}
\usepackage{dcolumn}
\usepackage{multirow}
\usepackage{booktabs}
\usepackage{enumitem}
\usepackage{graphicx}
\usepackage{tikz}
\usetikzlibrary{positioning,arrows.meta,calc,fit,backgrounds}
\usepackage{csquotes}
\usepackage{comment}
\usepackage{xcolor}
\usepackage[normalem]{ulem}
\usepackage{orcidlink}

\usepackage{xurl}
\usepackage{hyperref}
\hypersetup{
  colorlinks=true,
  citecolor=cyan,
  urlcolor=cyan,
  linkcolor=cyan,
  breaklinks=true,
}
\newtheoremstyle{break}
  {} {} {\itshape} {} {\bfseries} {.} {\newline} {\thmname{#1}\thmnumber{ #2}\thmnote{ (\bfseries #3)}}
\theoremstyle{break}
\newtheorem{theorem}{Theorem}
\newtheorem{mainresult}{Main Result}

\newtheorem{lemma}{Lemma}

\newtheorem{corollary}{Corollary}

\newtheorem{remark}{Remark}

\newenvironment{proofof}[1]{\begin{proof}[Proof of #1]}{\end{proof}}
\DeclareMathOperator{\Tr}{Tr}

\let\oldd\d \renewcommand{\d}{\ifmmode\mathrm{d}\else\oldd\fi}
\let\oldi\i \renewcommand{\i}{\ifmmode\mathrm{i}\else\oldi\fi}

\let\Log\relax \DeclareMathOperator{\Log}{Log}
\newcommand{\ket}[1]{| {#1} \rangle}
\newcommand{\bra}[1]{\langle {#1} |}
\newcommand{\braket}[1]{\langle #1 \rangle}
\newcommand{\ketbra}[2]{| {#1} \rangle\langle {#2} |}
\newcommand{\brakets}[3]{\langle {#1} | {#2} | {#3} \rangle}
\newcommand{\norm}[1]{\left\| {#1} \right\|}

\newcommand{\e}{\mathrm{e}}
\newcommand{\eps}{\epsilon}

\newcommand{\eq}[1]{\begin{align} #1 \end{align}}

\newcommand{\bbR}{\mathbb{R}}
\newcommand{\bbC}{\mathbb{C}}
\makeatletter
\@tfor\Letter:=ABCDEFGHIJKLMNOPQRSTUVWXYZ\do{%
  \expandafter\edef\csname rm\Letter\endcsname{\noexpand\mathrm{\Letter}}
  \expandafter\edef\csname bf\Letter\endcsname{\noexpand\mathbf{\Letter}}
  \expandafter\edef\csname sf\Letter\endcsname{\noexpand\mathsf{\Letter}}
  \expandafter\edef\csname cal\Letter\endcsname{\noexpand\mathcal{\Letter}}
  \expandafter\edef\csname scr\Letter\endcsname{\noexpand\mathscr{\Letter}}
}
\makeatother
\definecolor{lightblue}{rgb}{0.678, 0.847, 0.902}

\begin{document}

\title{A Path Integral Model of Cognition}

\author{Haruki Emori\,\orcidlink{0009-0007-2264-9192}}
\email{emori.haruki.i8@elms.hokudai.ac.jp}
\affiliation{Graduate School of Information Science and Technology, Hokkaido University, Kita 14, Nishi 9, Kita-ku, Sapporo, Hokkaido 060-0814, Japan}
\affiliation{RIKEN Center for Interdisciplinary Theoretical and Mathematical Sciences (iTHEMS), 2-1 Hirosawa, Wako, Saitama, 351-0198, Japan}
\affiliation{Teikyo University Advanced Comprehensive Research Organization (ACRO), 2-21-1 Kaga, Itabashi-ku, Tokyo, 173-0003, Japan}

\author{Kazunori Kondo\,\orcidlink{0009-0002-5798-2837}}
\email{kondo.hus@osaka-u.ac.jp}
\affiliation{Graduate School of Human Sciences, Department of Human Sciences, The University of Osaka, 1-2 Yamadaoka, Suita, Osaka, 565-0871, Japan}

\author{Atsushi Iriki\,\orcidlink{0000-0002-5262-163X}}
\email{iriki.atsushi.lh@teikyo-u.ac.jp}
\affiliation{Teikyo University Advanced Comprehensive Research Organization (ACRO), 2-21-1 Kaga, Itabashi-ku, Tokyo, 173-0003, Japan}
\affiliation{RIKEN Center for Interdisciplinary Theoretical and Mathematical Sciences (iTHEMS), 2-1 Hirosawa, Wako, Saitama, 351-0198, Japan}
\affiliation{Brain Mind and Consciousness Program, Canadian Institute for Advanced Research (CIFAR), 661 University Avenue, Suite 505, Toronto, ON M5G 1M1 Canada}

\author{Andrei Khrennikov\,\orcidlink{0000-0002-9857-0938}}
\email{andrei.khrennikov@lnu.se}
\affiliation{Center for Mathematical Modeling in Physics and Cognitive Sciences, Linnaeus University, V{\"a}xj{\"o}, SE-351 95, Sweden}

\date{\today}

\begin{abstract}
We develop the mathematical and physical formulation of cognitive cost optimization that underlies the path-integral model of consciousness.
The goal-directed cognitive process is modeled as imaginary-time evolution (ITE) under a projector Hamiltonian that rewards configurations consistent with a target concept.
We establish three results.
First, this ITE coincides with a double-bracket flow and is therefore the Riemannian gradient flow of a Hilbert--Schmidt cost whose unique minimum is the solution.
Second, a Wick rotation re-expresses this non-unitary descent as an equivalent unitary evolution on the same Hilbert space, which admits an exact discrete path-integral representation in which the oracle and the initial-state diffusion projector play the roles of potential and kinetic energy.
Third, we identify the continuum from unconscious to conscious processing with the strength of the unitary interaction between the cognitive system and a neural-environment probe, recovering the Gorini--Kossakowski--Sudarshan--Lindblad (GKSL) decoherence model of Asano \textit{et al.} in the Markovian weak-coupling limit, and the projective, reportable fixation of an optimized state in the strong-coupling limit, an insight-like ``Aha'' endpoint.
Both regimes share the same ITE and path-integral structure, and only the measurement-interaction strength varies.
The Wick rotation is therefore a technique of re-description, not a physical regime change.
\end{abstract}

\maketitle

\section{Introduction}
\label{sec:intro}
Contemporary consciousness research operates with a plurality of theories, among them Integrated Information Theory (IIT), the Global Neuronal Workspace (GNW), Recurrent Processing Theory (RPT), Higher-Order Thought (HOT), and Predictive Processing (PP), each offering a distinct account of how subjective experience arises in the brain~\cite{tononi2016,mashour2020,lamme2006,lau2011,friston2010}.

This paper is one component of a three-paper measurement-strength program.
Ref.~\cite{emori2026cc} introduces measurement strength as a coordinate axis for translating between IIT, GNW, RPT, HOT, PP, and related theories.
Ref.~\cite{emori2026ql} supplies the contextual and orthomodular foundation, showing how classical Boolean records arise when context is forgotten.
The present paper supplies the dynamical machinery, deriving how a cognitive state evolves, optimizes, and becomes weakly or strongly measured along that same axis.
A path integral model underlies the measurement-strength account of cognition and consciousness developed across the three papers.
The two formal elements of the continuum, the path integral and projective measurement~\cite{khrennikov2010,busemeyer2012,vonneumann1955,ozawa1984,ozawa2023}, are made rigorous and operative here through imaginary-time evolution (ITE) and the double-bracket flow (DBF).

For cognitive scientists, the path integral is a formal language for many possible cognitive histories rather than a claim about microscopic quantum processes in the brain.
Imaginary-time evolution and the double-bracket flow describe how this space of candidates is gradually biased toward a task-relevant or reward-consistent configuration.
Measurement strength then specifies how strongly the neural environment, the task, or the report demand reads out that evolving state.

A purely classical trajectory picture is not sufficient.
Classical optimization assumes that the alternatives, the cost function, and the sample space are fixed in advance.
Cognitive experiments often violate this assumption, since the order of questions matters, the task can change the state being measured, and a report can help determine the content it reports.
The present paper asks how such context-sensitive candidates evolve dynamically before one of them becomes a classical record.

Our object of study is the dynamical law that governs cognitive optimization.
Given an indeterminate initial mental state $\ket{\psi_{0}}$, a uniform superposition over conceptual elements, and an oracle of meaning $H_{f}$, a projector onto target concepts, we ask how the cognitive state evolves so as to minimize its divergence from the target.
Our central thesis is summarized as follows.
\begin{quote}
\itshape
Cognitive optimization, both unconscious and conscious, is governed by imaginary-time evolution (ITE) on a projector Hamiltonian.
The Wick rotation re-expresses the standard unitary description of the underlying system as an ITE, whose discrete path integral makes cost-minimizing trajectories computationally accessible.
The unconscious-to-conscious transition is a continuous deformation of the measuring-interaction strength between the cognitive system and its neural-environment probe.
\end{quote}
The dynamical claim of the present paper is that cognitive optimization is governed by a single path-integral structure, while the transition from unconscious candidate dynamics to conscious reportable fixation is controlled by the measuring-interaction strength between the cognitive system and its neural-environment probe.
This measuring-interaction strength is the dynamical realization of the measurement-strength axis used in Ref.~\cite{emori2026cc} to organize theories of consciousness.

Three threads substantiate this thesis.
The first is mathematical.
ITE on a projector Hamiltonian coincides with the double-bracket flow generated by a commutator, and is the Riemannian gradient flow of a least-squares Hilbert--Schmidt cost function~\cite{brockett1991,helmke2012,wiersema2023,gluza2026,mcmahon2025}.
This identifies the cost that ITE minimizes as the squared Hilbert--Schmidt distance from the target concept.
The second is physical.
Although ITE is formally non-unitary, it is connected, through the Wick rotation $t = -\i\tau$, to a unitary dynamics on the same Hilbert space generated by the effective Hamiltonian $\calH_{\text{eff}} = \i[H_{f},\psi_{0}]$.
Following recent work on double-bracket quantum algorithms~\cite{gluza2024,gluza2026,gilyen2019,martyn2021}, this unitary admits an \textit{exact} discrete path-integral representation on a finite-dimensional Hilbert space, with no continuum limit and no measure-theoretic ambiguity.
The third is foundational.
The open-system reformulation of cognitive dynamics via Gorini--Kossakowski--Sudarshan--Lindblad (GKSL)-type decoherence~\cite{gorini1976,lindblad1976} introduced by Asano, Ohya, Tanaka, Basieva, and Khrennikov~\cite{asano2011} interprets the gradual emergence of a decision as a sequence of weak measurements, which we recover here, under a Markovian idealization, as the open-system reduction of the discrete cognitive path integral.

For cognitive-science readers, the formalism has four intended contact points.
First, the finite Hilbert space can be read as a representational space, including distributed vector representations~\cite{piantadosi2024}.
Second, the Hilbert--Schmidt cost plays a role analogous to a task objective or a self-prior~\cite{kim2025}.
Third, weak monitoring by the neural environment yields open-system dynamics related to decoherence and to nonequilibrium signatures~\cite{nartallo2026}.
Fourth, embodied action can be interpreted as pruning the space of virtual cognitive trajectories~\cite{zou2025}.

The organizing distinction of our framework is between mathematical \textit{description} and physical \textit{process}.
The Wick rotation is not a physical transition from one cognitive regime to another.
It is a mathematical bridge between two equivalent descriptions of the same dynamical content, ITE on the cost function on the one hand and unitary evolution on the same Hilbert space on the other.
Both unconscious and conscious cognitive processing is described, at the level of dynamics, by the same ITE flow on the projector Hamiltonian $H_{f}$.
What differentiates them is not the form of the dynamical law but the \textit{strength of the measurement coupling} between the cognitive system and the neural environment that acts as its probe.
These three threads fix the role of the present paper within the broader program.
Ref.~\cite{emori2026cc} uses the same measurement-strength axis as an empirical and theoretical map, Ref.~\cite{emori2026ql} explains why context-sensitive, non-Boolean structure is required, and the present paper shows how candidate histories, cost descent, and system--probe measurement strength are generated dynamically.

We adopt the indirect measurement model~\cite{vonneumann1955,ozawa1984,ozawa2023}, in which every physical measurement is realized through a unitary interaction $U_{\text{int}}$ between the system and a probe, followed by a projective readout of the probe.
The strength of the unitary interaction controls whether the resulting measurement is \textit{weak}, with small disturbance, partial information, and preserved interference, or \textit{strong}, with large disturbance, complete information, and projective collapse.
The unconscious regime corresponds to weak coupling, where the cognitive state evolves under ITE while interference among virtual mental trajectories is preserved, and outcomes are averaged to yield, under the additional Markovian idealization of independent and identically prepared probes, the GKSL-type dynamics of Ref.~\cite{asano2011}.
The conscious regime corresponds to strong coupling, where the same ITE evolution is sharply read out by the environment, collapsing the optimized superposition into a single, reportable thought (an insight-like ``Aha'' moment).
Between these extremes lies an intermediate regime that interpolates continuously between the two, recovering the measurement-strength continuum used to organize theories of consciousness in Ref.~\cite{emori2026cc}.

This paper is organized as follows.
Section~\ref{sec:preliminary} reviews the interplay between ITE, DBF, and Riemannian gradient flow, which provides the geometric basis for the model.
Section~\ref{sec:path-integral} reviews the Feynman path integral and the Wick rotation in their standard form, in order to fix vocabulary and to make explicit the correspondence the cognitive model uses.
Section~\ref{sec:pim} contains the main results.
It derives an exact discrete path integral for the unitary form of cognitive ITE (Main Result~\ref{mr:exact_path_integral}), establishes a structural isomorphism with the Feynman path integral, expands the Hilbert--Schmidt cost function to show why ITE drives the state to the solution subspace, and develops the measurement-strength continuum that links the present work to Refs.~\cite{emori2026cc,emori2026ql}.
Section~\ref{sec:conclusion} concludes.
Standard terminology, the relation between Grover's algorithm and ITE, and the geodesic structure of the resulting trajectory are deferred to Appendices~\ref{sec:terminology}--\ref{sec:geo}.
As a reading guide, readers mainly interested in consciousness science may read Secs.~\ref{sec:meas_strength}--\ref{sec:lindblad} and the Conclusion as the empirical interface of the paper, since these sections identify the coupling strength with the measurement-strength variable manipulated by report demand, task pressure, and environmental readout.
Readers interested in the formal derivation should read Secs.~\ref{sec:preliminary}--\ref{sec:hs_expansion} and the appendices.

We close this introduction with two cautions, both inherited from Refs.~\cite{emori2026cc,emori2026ql}.
First, our use of quantum vocabulary is mathematical, not ontological, and we do \textit{not} claim that the brain is a quantum physical device at biological scales.
We claim that cognitive state transitions exhibit the algebraic non-commutativity, contextual probability-space construction, and observation-induced state change for which the Hilbert-space formalism was designed, and that the ITE and path-integral picture together with the measurement-strength continuum provides a common language for these features.
Second, the path-integral derivation we present is an algebraic identity on a finite-dimensional Hilbert space.
It does not invoke any continuum limit, and it should not be confused with the formal manipulations of measure-zero functional integrals to which the standard Feynman construction is sometimes pushed~\cite{glimm1987}.
The exactness of this representation is what makes the framework implementable, in particular on near-term quantum hardware through the quantum singular value transformation (QSVT)~\cite{gilyen2019,martyn2021,suzuki2025} and the double-bracket quantum algorithm (DBQA)~\cite{gluza2024,gluza2026}.

\section{Dynamical Foundations: Imaginary-Time Evolution, Double-Bracket Flow, and Gradient Descent}
\label{sec:preliminary}
We collect the dynamical objects on which the model is built.
The chain of equivalences established here,
\eq{
\text{ITE} \;\Longleftrightarrow\; \text{DBF on the orbit manifold} \;\Longleftrightarrow\; \text{Riemannian gradient flow of an HS cost,}
}
is the geometric basis for the cognitive optimization dynamics of Sec.~\ref{sec:pim}, and is summarized in Fig.~\ref{fig:equivalence}.
No phenomenological claim is attached to imaginary time here.
The variable $\tau$ is not a hidden mental clock but a coordinate that recasts the same optimization dynamics in a form where candidate histories and cost weighting can be compared.
Additional norms, operator-theoretic conventions, and convergence lemmas are gathered in Appendix~\ref{sec:terminology}.

\begin{figure}[t]
\centering
\begin{tikzpicture}[>=Latex, every node/.style={font=\small}]
\tikzstyle{box}=[draw, rounded corners, align=center, minimum height=10mm, inner sep=5pt]
\node[box] (ite) {Imaginary-time\\ evolution (ITE)\\ $\e^{-\tau H}$};
\node[box, right=20mm of ite] (dbf) {Double-bracket\\ flow (DBF)\\ $\partial_\tau\Psi=[[\Psi,H],\Psi]$};
\node[box, right=20mm of dbf] (rgf) {Riemannian gradient\\ flow of the\\ Hilbert--Schmidt cost};
\draw[<->,thick] (ite) -- (dbf);
\draw[<->,thick] (dbf) -- (rgf);
\node[box, fill=lightblue!40, below=12mm of dbf] (sol) {unique minimum\\ $=$ solution state $\ket{\psi^{*}}$};
\draw[->,thick] (rgf.south) |- (sol.east);
\draw[->,thick] (ite.south) |- (sol.west);
\end{tikzpicture}
\caption{\label{fig:equivalence}The equivalence used throughout the paper.
Imaginary-time evolution of a pure state, the double-bracket flow of its density matrix, and the Riemannian gradient flow of a Hilbert--Schmidt cost function coincide, and all three converge to the state that minimizes the cost.
In the cognitive reading of Sec.~\ref{sec:pim}, this minimum is the configuration consistent with the target concept.}
\end{figure}
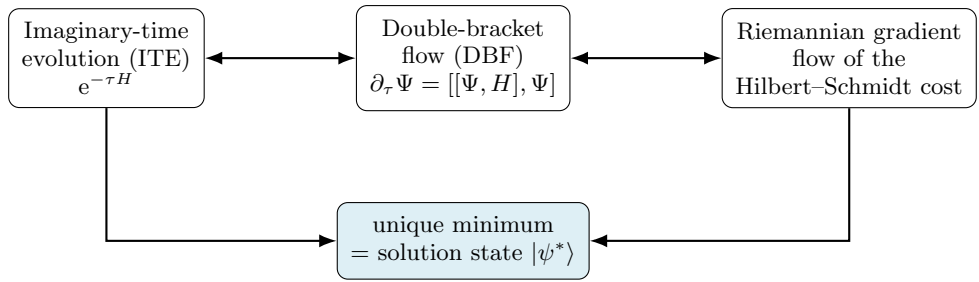

\subsection{Double-Bracket Flow as Riemannian Gradient Flow}
\label{sec:dbf_rgf}
We review DBF, a class of matrix-valued ordinary differential equations introduced by Brockett as a steepest-descent method for matrix diagonalization, QR decomposition, and eigenvalue sorting~\cite{brockett1991,deift1983,chu1988,helmke2012}.
Consider a self-adjoint operator $A$ and the unitary orbit manifold
\eq{
\calM(A) = \{U A U^{\dagger} \mid U U^{\dagger} = I\}.
}
For a second self-adjoint operator $B$, define the least-squares cost function
\eq{
C_{B}(X) = -\tfrac{1}{2}\norm{X - B}_{\text{HS}}^{2},\qquad X \in \calM(A),\label{eq:Cost_DBF}
}
where $\norm{\cdot}_{\text{HS}}$ is the Hilbert--Schmidt norm.
Its Riemannian gradient at $X$ is~\cite{helmke2012,wiersema2023}
\eq{
\mathrm{grad}_{X} C_{B}(X) = [X,[X, B]].\label{eq:rg}
}
The derivation of Eq.~\eqref{eq:rg}, summarized in Appendix~\ref{sec:terminology}, rests on the tangency condition $T_{X}\calM(A) = \{[X,\xi] \mid \xi^{\dagger} = -\xi\}$ and the compatibility of the Riemannian metric $\braket{[X,\Omega_1],[X,\Omega_2]} = \Tr[\Omega_1^{\dagger}\Omega_2]$ with the directional derivative of \eqref{eq:Cost_DBF}.
The gradient-descent flow $\partial A(t)/\partial t = -\mathrm{grad}_{A(t)} C_{B}(A(t))$ is the \textit{double-bracket flow}
\eq{
\frac{\partial A(t)}{\partial t} = [[A(t), B], A(t)].\label{eq:dbf}
}

\subsection{Equivalence of ITE and Double-Bracket Flow}
\label{sec:ite_dbf}
Imaginary-time evolution of a pure state is an instance of DBF.
For a Hamiltonian $H$, the normalized ITE state at imaginary time $\tau \in \bbR$ is
\eq{
\ket{\Psi(\tau)} = \frac{\e^{-\tau H}\ket{\Psi(0)}}{\norm{\e^{-\tau H}\ket{\Psi(0)}}}.
}
Differentiating with respect to $\tau$ yields the imaginary-time Schr\"{o}dinger equation
\eq{
\frac{\partial \ket{\Psi(\tau)}}{\partial \tau}
= -\bigl(H - E(\tau)\,I\bigr)\ket{\Psi(\tau)}
= -[H, \Psi(\tau)]\ket{\Psi(\tau)},
}
with $E(\tau) = \brakets{\Psi(\tau)}{H}{\Psi(\tau)}$.
Promoting to the density matrix $\Psi(\tau) = \ketbra{\Psi(\tau)}{\Psi(\tau)}$,
\eq{
\frac{\partial \Psi(\tau)}{\partial \tau} = [[\Psi(\tau), H], \Psi(\tau)],\label{eq:ite_dbf}
}
which is the DBF \eqref{eq:dbf} with $A = \Psi(\tau)$ and $B = H$.
Substituting into \eqref{eq:Cost_DBF},
\eq{
C_{H}(\Psi(\tau))
&= -\tfrac{1}{2}\norm{\Psi(\tau)-H}_{\text{HS}}^{2} \nonumber\\
&= \underbrace{-\tfrac{1}{2}\norm{\Psi(\tau)}_{\text{HS}}^{2} - \tfrac{1}{2}\norm{H}_{\text{HS}}^{2}}_{\text{constants in }\tau} + \underbrace{\Tr[\Psi(\tau) H]}_{=\,E(\tau)}.
\label{eq:ite-cost}
}
The first two terms are conserved along the flow, so minimizing the HS cost $C_{H}$ is equivalent to minimizing the energy, and ITE is the Riemannian gradient-descent flow of $C_{H}$.

\subsection{Application to Unstructured Search and Cognitive Concept Search}
\label{sec:ite_search}
The same formalism applies, with a sign reversal, to unstructured search.
Let $H_{f} = \sum_{x:f(x)=1}\ketbra{x}{x}$ be the projector onto the marked subspace and $\ket{\psi_0}$ the uniform superposition.
The positive-sign ITE
\eq{
\ket{\Phi(\tau)} = \frac{\e^{\tau H_{f}}\ket{\psi_0}}{\sqrt{\brakets{\psi_0}{\e^{2\tau H_{f}}}{\psi_0}}}
}
satisfies
\eq{
\frac{\partial \Phi(\tau)}{\partial \tau} = [[H_{f},\Phi(\tau)],\Phi(\tau)],
}
which is the Riemannian gradient flow of the \textit{ascending} HS cost
\eq{
\widetilde{C}_{H_{f}}(\Psi) = +\tfrac{1}{2}\norm{H_{f}-\Psi}_{\text{HS}}^{2}.\label{eq:g-cost}
}
The sign difference between Eqs.~\eqref{eq:ite-cost} and \eqref{eq:g-cost} reflects only that standard ITE seeks the smallest eigenvalue while unstructured search seeks the largest, and both fit one template by choosing the sign of $H$ appropriately.
A more detailed analysis of this case, including its identification with a geodesic on $\bbC P^{N-1}$, is given in Appendix~\ref{sec:geo}.
In the cognitive interpretation that we develop in Sec.~\ref{sec:pim}, $H_{f}$ is the \textit{oracle of meaning}, a self-adjoint operator that assigns a high energy (reward) to mental configurations consistent with the target concept and zero otherwise.

\section{Path Integral and the Wick Rotation}
\label{sec:path-integral}
To place the cognitive model in its mathematical context, we review the Feynman path integral and the Wick rotation~\cite{feynman1948,dyson1949,feynman1949}.
This section serves two roles.
First, it fixes vocabulary for Sec.~\ref{sec:pim}, where a discrete cognitive path integral is derived.
Second, it makes precise the mathematical content of the analytic continuation $t \to -\i\tau$ that, in Sec.~\ref{sec:exact_pi}, supplies the bridge between the unitary description of a system and the ITE description of cost optimization.
In our framework this Wick rotation is a technique of description, a re-encoding of the same physical content, and not a transition between distinct physical regimes.

For a time-independent Hamiltonian $H = T + V$ with $[T,V] \neq 0$, the unitary propagator $\exp(-\i H t/\hbar)$ cannot be factorized exactly into individual exponentials.
The first-order Trotter--Suzuki product formula approximates
\eq{
\e^{-\i(T+V)t/\hbar} = \lim_{n\to\infty}\left(\e^{-\i T \Delta t /\hbar}\e^{-\i V \Delta t /\hbar}\right)^{n},
\qquad \Delta t = t/n,
}
with $\calO(\Delta t^{2})$ error per step.
Inserting position-basis resolutions of the identity into the resulting product, one obtains
\eq{
K(x_f,t_f;x_i,t_i) = \bra{x_f}\e^{-\i H(t_f - t_i)/\hbar}\ket{x_i}
= \int\!\!\d x_{n-1}\cdots\!\int\!\!\d x_{1}\prod_{j=1}^{n}\brakets{x_j}{\e^{-\i T \Delta t/\hbar}\e^{-\i V \Delta t/\hbar}}{x_{j-1}}.
}
Each single-step amplitude evaluates to
\eq{
\brakets{x_{j}}{\e^{-\i T\Delta t/\hbar}\e^{-\i V\Delta t/\hbar}}{x_{j-1}}
= \sqrt{\frac{m}{2\pi\i\hbar\Delta t}}\exp\!\left\{\frac{\i}{\hbar}\!\left[\frac{m}{2}\!\left(\frac{x_j - x_{j-1}}{\Delta t}\right)^{2} - V(x_{j-1})\right]\!\Delta t\right\},
}
so that in the limit $n\to\infty$,
\eq{
K(x_f,t_f;x_i,t_i) = \int\!\calD[x(t)]\,\e^{\i S[x(t)]/\hbar},\qquad
S[x] = \int_{t_i}^{t_f}\!\calL(x,\dot{x},t)\,\d t.
\label{eq:feynman_pi}
}
The classical limit $\hbar \to 0$ is recovered as a stationary-phase principle, since only paths near $\delta S = 0$ contribute constructively.
The Wick rotation $t \to -\i\tau$ converts \eqref{eq:feynman_pi} into the Euclidean path integral
\eq{
K_{E}(x_f,\tau;x_i,0) = \int\!\calD[x(\tau)]\,\e^{-S_{E}[x(\tau)]/\hbar},\label{eq:euclidean_pi}
}
in which the oscillatory phase becomes a real, exponentially decaying weight.
Two consequences will be relevant in Sec.~\ref{sec:pim}.
First, the Euclidean propagator coincides with the canonical-ensemble partition function $Z = \Tr[\e^{-\beta H}]$ under the identification $\hbar\beta = \tau$, so imaginary-time quantum dynamics and thermal statistical mechanics share the same Boltzmann weight.
Second, the Euclidean weight is integrable in the strict mathematical sense and admits a probability interpretation, and correlation functions are generated by the Euclidean generating functional~\cite{schwinger1951a,schwinger1951b,schwinger1951c,schwinger1953,emori2026qsf},
\eq{
Z_{E}[J] = \int\!\calD\phi\;\e^{-S_{E}[\phi] + \braket{J,\phi}}.
}
The price is the loss of manifest quantum interference, since the complex amplitudes of \eqref{eq:feynman_pi} are traded for positive Boltzmann weights.

The above is the standard physical narrative.
In Sec.~\ref{sec:pim} we use this duality in a specific direction.
We start from a non-unitary ITE that captures cognitive cost optimization, use the Wick rotation $t = -\i\tau$ to re-express it as an equivalent unitary dynamics on the same Hilbert space, and then construct an exact discrete path integral for that unitary on a finite-dimensional Hilbert space.
In this construction the Wick rotation is the technical step that allows the Euclidean ITE description, in which the content is the descent of a cost function, and the Minkowski path-integral description, in which the content is the interference of trajectories, to be used interchangeably as descriptions of the same dynamics.

\section{Main Result: Dynamics along Measurement Strength as a Path Integral}
\label{sec:pim}
We now state the main results.
Section~\ref{sec:exact_pi} derives an exact discrete path integral for the unitary form of cognitive ITE on a finite-dimensional Hilbert space, separating the algebraically exact path-integral identity from the approximate product-formula representation of the cognitive unitary.
Section~\ref{sec:isomorphism} establishes the structural isomorphism between this path integral and the Feynman construction of Sec.~\ref{sec:path-integral}.
Section~\ref{sec:hs_expansion} expands the HS cost function and shows why ITE drives the cognitive state to the marked subspace, with a convergence rate set by the spectral gap.
Section~\ref{sec:meas_strength} states the main interpretive result, namely that both unconscious and conscious cognitive dynamics share the same ITE and path-integral structure, and that the unconscious-to-conscious transition is a continuous deformation of the measurement-interaction strength.
Section~\ref{sec:lindblad} connects the Markovian weak-coupling limit of the measuring interaction to the GKSL-type decoherence framework of Asano \textit{et al.}~\cite{asano2011}.
Figure~\ref{fig:strength} summarizes the measurement-strength continuum that these subsections develop.

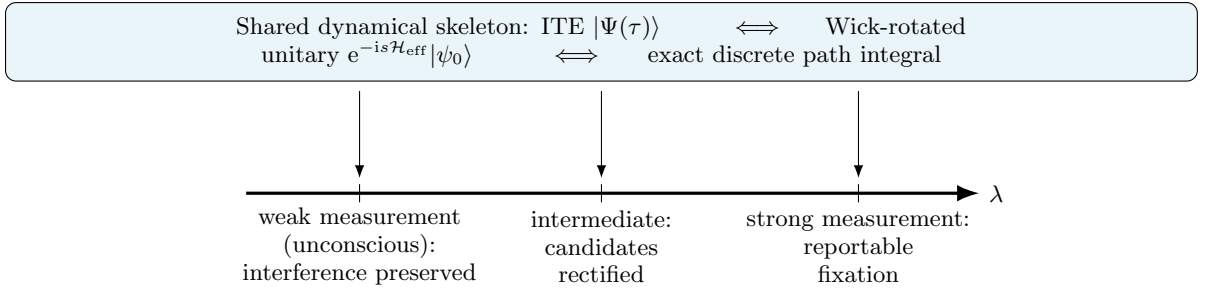
\begin{figure}[t]
\centering
\begin{tikzpicture}[>=Latex, every node/.style={font=\small}]
\node[draw, rounded corners, fill=lightblue!25, align=center, text width=0.86\linewidth, inner sep=5pt] (skel) at (0,2)
{Shared dynamical skeleton: ITE $\ket{\Psi(\tau)} \;\Longleftrightarrow\;$ Wick-rotated unitary $\e^{-\i s\calH_{\mathrm{eff}}}\ket{\psi_0} \;\Longleftrightarrow\;$ exact discrete path integral};
\draw[->,very thick] (-4.7,0) -- (5.0,0) node[right]{$\lambda$};
\foreach \x in {-3.2, 0, 3.4} {\draw (\x,0.12)--(\x,-0.12);}
\node[align=center] at (-3.2,-0.7) {weak measurement\\ (unconscious):\\ interference preserved};
\node[align=center] at (0,-0.7) {intermediate:\\ candidates\\ rectified};
\node[align=center] at (3.4,-0.7) {strong measurement:\\ reportable\\ fixation};
\draw[->] (-3.2,1.35) -- (-3.2,0.2);
\draw[->] (0,1.35) -- (0,0.2);
\draw[->] (3.4,1.35) -- (3.4,0.2);
\end{tikzpicture}
\caption{\label{fig:strength}Cognitive dynamics along the measurement-strength axis.
The same dynamical skeleton, the ITE flow of Eq.~\eqref{eq:ite_norm}, its Wick-rotated unitary form of Eq.~\eqref{eq:Ueff}, and the exact discrete path integral of Main Result~\ref{mr:exact_path_integral}, governs cognition at every value of the system--probe coupling strength $\lambda$.
Increasing $\lambda$ moves the induced measurement from weak, where interference among candidate histories is preserved and no outcome is reportable, through an intermediate regime in which candidates are rectified, to strong, where a projective measurement produces a reportable outcome.
The parameter $\lambda$ is the dynamical realization of the measurement-strength axis of Ref.~\cite{emori2026cc}.}
\end{figure}

\subsection{The Unitary Form of Cognitive ITE and Its Exact Discrete Path Integral}
\label{sec:exact_pi}
We model the unconscious search for a target concept as a dynamical process on a finite-dimensional Hilbert space $\calH \cong \bbC^{N}$ spanned by an orthonormal computational basis $\{\ket{x}\}_{x=0}^{N-1}$ of conceptual elements.
The word ``conceptual'' should be read broadly, since the basis may represent symbolic alternatives, distributed semantic vectors~\cite{piantadosi2024}, discretized regions of a neural state manifold, or experimentally reconstructed representational states, where geometric proximity encodes semantic similarity and algebraic operations capture compositional structure.
Two operators play a distinguished role.
\begin{enumerate}[label=(\roman*)]
  \item the \textit{oracle of meaning} $H_{f} = \sum_{x:f(x)=1}\ketbra{x}{x}$, a projector that assigns a unit reward to a mental configuration consistent with the target concept and zero otherwise;
  \item the \textit{initial-state diffusion operator} $\psi_{0} = \ketbra{\psi_0}{\psi_0}$, where $\ket{\psi_0} = N^{-1/2}\sum_{x}\ket{x}$ is the uniform superposition over conceptual elements, representing the maximally indeterminate cognitive state.
\end{enumerate}
In cognitive terms, $H_{f}$ need not be an oracle in the literal computational sense.
It can represent a task goal, a learned category, a self-prior, or a reward-consistent region of representational space, and the initial state $\psi_0$ represents not ignorance but an unresolved field of candidate interpretations.

\paragraph*{From ITE to a commutator unitary via Wick rotation.}
By Sec.~\ref{sec:ite_search}, the cognitive cost minimization is the gradient flow~\eqref{eq:g-cost}, equivalently the ascending ITE on $H_{f}$,
\eq{
\ket{\Psi(\tau)} = \frac{\e^{\tau H_{f}}\ket{\psi_{0}}}{\norm{\e^{\tau H_{f}}\ket{\psi_{0}}}}.\label{eq:ite_norm}
}
A standard fact, made precise in Lemma~\ref{lem:ite_commutator} (Appendix~\ref{sec:ga=ite}), is that for the projector $H_{f}$ the ITE state~\eqref{eq:ite_norm} admits an \textit{exact} representation as the action of a unitary generated by a commutator,
\eq{
\ket{\Psi(\tau)} = \e^{s_{\tau}[H_{f},\psi_{0}]}\ket{\psi_{0}},\label{eq:ite_commutator}
}
for a reparameterized duration $s_{\tau}$ given in Eq.~\eqref{eq:opt_s}.
We interpret \eqref{eq:ite_commutator} as the Wick-rotated form of the ITE~\eqref{eq:ite_norm}.
Since $[H_{f},\psi_{0}]^{\dagger} = -[H_{f},\psi_{0}]$, the operator
\eq{
\calH_{\text{eff}} \equiv \i[H_{f},\psi_{0}]
\label{eq:Heff_def}
}
is Hermitian, and we may write
\eq{
\ket{\psi_{s}} = \e^{s[H_{f},\psi_{0}]}\ket{\psi_{0}} = \e^{-\i s\calH_{\text{eff}}}\ket{\psi_{0}}.\label{eq:Ueff}
}
The right-hand side is the unitary description of a Hermitian Hamiltonian acting for a real ``time'' $s$, and the left-hand side coincides, by~\eqref{eq:ite_commutator}, with the normalized ITE state at imaginary time $\tau(s)$.
The mapping $\tau \leftrightarrow s$ furnished by the inverse of Eq.~\eqref{eq:opt_s} is the content of the Wick rotation in this setting, identifying the descent of a cognitive cost function with a unitary trajectory generated by $\calH_{\text{eff}}$.

\paragraph*{Discrete approximation of the cognitive unitary.}
Following the DBQA construction~\cite{gluza2024,gluza2026} (see Appendix~\ref{sec:ga=ite}), we discretize $U(s) = \e^{-\i s\calH_{\text{eff}}}$ using $\calN$ interleaved product-formula steps,
\eq{
U(s) \;\approx\; U_{\calN} \;\equiv\; (-1)^{\calN}\prod_{k=1}^{\calN} G_{k}(\alpha_k,\beta_k),
\qquad G_{k}(\alpha_k,\beta_k) = \e^{\i\alpha_k\psi_0}\,\e^{\i\beta_k H_{f}},
\label{eq:UN}
}
where the phase angles $\{\alpha_k,\beta_k\}$ are determined by the chosen order of the product formula.
For the original Grover schedule $\alpha_k = \beta_k = \pi$, and for fixed-point search they follow recursive quasi-Chebyshev schedules~\cite{yoder2014,li2026}.
The accuracy of $U(s) \approx U_{\calN}$ depends on the product-formula order, a controlled approximation in the standard sense of Hamiltonian simulation~\cite{chen2022,childs2013,dawson2006}.
The following statement holds \textit{for any choice of phase angles}, since the matrix elements of $U_{\calN}$ admit an exact, rather than approximate, discrete path-integral representation.
The approximation, if any, lies in the relation between $U_{\calN}$ and the target unitary $U(s)$, while the path integral itself is an algebraic identity.

\paragraph*{Insertion of completeness and exact local amplitudes.}
Inserting the resolution of identity $I = \sum_{x}\ketbra{x}{x}$ between every successive factor in \eqref{eq:UN} expresses the discrete propagator as a sum over all $\calN$-step trajectories $x = (x_0,x_1,\dots,x_{\calN})$ in conceptual space,
\eq{
K_{\calN}(x_f,x_i) = \bra{x_f}U_{\calN}\ket{x_i}
= (-1)^{\calN}\!\!\!\sum_{x_1,\dots,x_{\calN-1}}\prod_{k=1}^{\calN}\brakets{x_k}{\e^{\i\alpha_k\psi_0}\e^{\i\beta_k H_{f}}}{x_{k-1}},\label{eq:KN_sum}
}
with $x_0 = x_i,\,x_{\calN} = x_f$.
We evaluate each local factor exactly.
Since $H_{f}\ket{x} = f(x)\ket{x}$ with $f(x) \in \{0,1\}$,
\eq{
\e^{\i\beta_k H_{f}}\ket{x_{k-1}} = \e^{\i\beta_k f(x_{k-1})}\ket{x_{k-1}}.
}
The operator $\e^{\i\alpha_k\psi_0}$ acts on the rank-one projector $\psi_0$.
Using $\psi_0^{2} = \psi_0$, the Taylor series terminates,
\eq{
\e^{\i\alpha_k\psi_0} = I + (\e^{\i\alpha_k}-1)\psi_0,
}
and its matrix elements read
\eq{
\brakets{x_k}{\e^{\i\alpha_k\psi_0}}{x_{k-1}} = \delta_{x_k,x_{k-1}} + \frac{\e^{\i\alpha_k}-1}{N}.\label{eq:diffusion_matrix}
}
Combining \eqref{eq:diffusion_matrix} with the oracle factor, the local transition amplitude becomes
\eq{
\brakets{x_k}{G_{k}(\alpha_k,\beta_k)}{x_{k-1}}
= \exp\!\left\{\i\!\left[\beta_k f(x_{k-1}) \;-\; \i\Log\!\left(\delta_{x_k,x_{k-1}} + \tfrac{\e^{\i\alpha_k}-1}{N}\right)\right]\right\},\label{eq:local_G}
}
where $\Log$ denotes the principal branch of the complex logarithm.
Two algebraic observations make \eqref{eq:local_G} unambiguous.
For $\alpha_k \in \bbR$ the argument of the logarithm is nonzero whenever either $\delta_{x_k,x_{k-1}}=1$, giving $1 + N^{-1}(\e^{\i\alpha_k}-1)$ with magnitude bounded below by $1 - 2/N$, or $\delta_{x_k,x_{k-1}}=0$ and $\e^{\i\alpha_k}\neq 1$, giving $N^{-1}(\e^{\i\alpha_k}-1)$.
The degenerate case $\alpha_k \equiv 0\;(\mathrm{mod}\,2\pi)$ collapses $G_{k}$ to a diagonal oracle factor and the path integral reduces accordingly.
Taking the product over all $\calN$ steps yields the central technical result of this paper.
\begin{mainresult}[Exact discrete cognitive path integral]\label{mr:exact_path_integral}
Let $H_{f}$ be a projector and $\ket{\psi_0}$ the uniform superposition on $\calH = \bbC^{N}$, and let $\{\alpha_k,\beta_k\}_{k=1}^{\calN} \subset \bbR$ be arbitrary phase angles with $\alpha_k \not\equiv 0\;(\mathrm{mod}\,2\pi)$.
The discrete propagator \eqref{eq:KN_sum} admits the \textit{exact} representation
\eq{
K_{\calN}(x_f,x_i) = (-1)^{\calN}\sum_{\mathrm{paths}\;x}\e^{\i S_{\mathrm{eff}}[x]},
\label{eq:KN_exact}
}
with the discrete cognitive action
\eq{
S_{\mathrm{eff}}[x] = \sum_{k=1}^{\calN}\!\left[\beta_k f(x_{k-1}) \;-\; \i\Log\!\left(\delta_{x_k,x_{k-1}} + \frac{\e^{\i\alpha_k}-1}{N}\right)\right].\label{eq:discrete_action}
}
\end{mainresult}
The representation \eqref{eq:KN_exact} follows by direct computation, since Eqs.~\eqref{eq:Ueff}--\eqref{eq:local_G} are algebraic equalities on the finite-dimensional Hilbert space $\bbC^{N}$.
No continuum limit ($\Delta t \to 0$) is invoked, and the multi-dimensional integral over $\d x_1\cdots\d x_{n-1}$ that complicates the measure-theoretic interpretation of the continuum Feynman integral is here replaced by a finite sum.
\begin{remark}[Two layers of approximation]
The relation $U_{\calN} \approx U(s)$ in \eqref{eq:UN} is an approximation in the standard sense of product-formula Hamiltonian simulation, with accuracy controlled by the order of the schedule and the duration $s$~\cite{chen2022,childs2013}.
The identity \eqref{eq:KN_exact} is, by contrast, exact for the given $U_{\calN}$ regardless of its accuracy as an approximation of $U(s)$.
This separation is what keeps the path-integral representation algebraically clean and sidesteps the measure-theoretic ambiguity of the continuum functional integral~\cite{glimm1987}.
\end{remark}

\subsection{Structural Isomorphism with the Feynman Path Integral}
\label{sec:isomorphism}
The discrete action \eqref{eq:discrete_action} corresponds to the discretized Feynman action of Sec.~\ref{sec:path-integral}.
In continuous physics, Trotterization accommodates the non-commutativity $[T,V]\neq 0$, and in our model the product formula resolves the non-commutativity $[\psi_0,H_{f}]\neq 0$.
This delivers the dictionary summarized in Table~\ref{tab:isomorphism}.
\begin{table}[t]
\caption{\label{tab:isomorphism}Structural isomorphism between the Feynman and cognitive path integrals.
The cognitive/consciousness reading gives the intended interpretation of each row for readers in the cognitive sciences.
Continuum position space corresponds to physical trajectories before observation, discrete conceptual space to candidate interpretations or representational states, the non-commutativity $[\psi_0,H_f]\neq 0$ to the fact that task and current state are not jointly fixed, the oracle $H_f$ to a goal, category, self-prior, or reward-consistent target, the diffusion $\psi_0$ to diffusion over candidate interpretations, the action weight $\exp(\i S_{\mathrm{eff}}[x])$ to interference among possible cognitive histories, and the stationary phase to the dominance of cost-minimizing interpretations.}
\centering
\begin{tabular}{lll}
\hline\hline
& \textit{Feynman path integral} & \textit{Cognitive path integral} \\
\hline
Carrier space & Continuum position space $\bbR^{d}$ & Discrete conceptual space $\{0,\dots,N-1\}$ \\
Non-commutativity & $[T,V]\neq 0$ & $[\psi_0,H_{f}]\neq 0$ \\
``Potential energy'' & $V(x)$ & $H_{f}$ (oracle of meaning) \\
``Kinetic energy'' & $T = p^{2}/2m$ & $\psi_0$ (initial-state diffusion) \\
Local move & Gaussian Brownian step (free-particle propagator) & $\delta_{x_k,x_{k-1}} + (\e^{\i\alpha_k}-1)/N$ \\
Action weight & $\exp(\i S[x]/\hbar)$ & $\exp(\i S_{\mathrm{eff}}[x])$ \\
Classical limit & Stationary phase $\delta S = 0$ & Stationary phase under cost minimization \\
\hline\hline
\end{tabular}
\end{table}
The mapping is concrete in two senses.
The oracle $H_{f}$ defines a potential landscape over conceptual elements just as $V(x)$ defines one over physical positions, contributing a phase $\beta_k f(x_{k-1})$ each time the trajectory visits a marked concept.
The projector $\psi_0$ acts as a uniformly diffusing kernel, and its matrix elements \eqref{eq:diffusion_matrix} permit either remaining at the current concept, through $\delta_{x_k,x_{k-1}}$, or jumping uniformly to any other concept, through $(\e^{\i\alpha_k}-1)/N$, in analogy to the Gaussian transition kernel produced in the continuum case by momentum integration of $p^{2}/(2m)$.
This identifies the oracle and diffusion operators of Grover-type search~\cite{grover1996} and of DBQA~\cite{gluza2024,gluza2026} as the cognitive counterparts of potential and kinetic energy.

The discrete cognitive action $S_{\mathrm{eff}}[x]$ also admits a reading in terms of cognitive load~\cite{zou2025}.
Trajectories that wander through conceptual space accumulate large action phases, which is analogous to high extraneous cognitive load.
Embodied cognition, such as gesture or environmental manipulation, can then be understood as a mechanism that biases the diffusion kernel $\psi_0$, so that offloading computational demands to the environment prunes irrelevant conceptual trajectories, reduces extraneous load, and favors the interference of cost-minimizing paths.

\subsection{Why ITE Settles on the Solution: Expansion of the Hilbert--Schmidt Cost}
\label{sec:hs_expansion}
We make explicit the mechanism by which the gradient flow of \eqref{eq:g-cost} converges to the target.
Writing $\Psi = \ketbra{\psi}{\psi}$ for an instantaneous pure state of the cognitive system,
\eq{
\widetilde{C}_{H_{f}}(\Psi)
&= \tfrac{1}{2}\norm{H_{f}-\Psi}_{\text{HS}}^{2}
= \tfrac{1}{2}\Tr\!\left[(H_{f}-\Psi)^{\dagger}(H_{f}-\Psi)\right]\nonumber\\
&= \tfrac{1}{2}\Tr[H_{f}^{2}] + \tfrac{1}{2}\Tr[\Psi^{2}] - \Tr[H_{f}\Psi].
\label{eq:HS_expand_1}
}
Two simplifications follow from the structure of the problem.
First, $H_{f}$ is a projector onto the $M$-dimensional marked subspace, so $H_{f}^{2} = H_{f}$ and $\Tr[H_{f}^{2}] = \Tr[H_{f}] = M$.
Second, $\Psi$ is a pure-state density operator, so $\Psi^{2} = \Psi$ and $\Tr[\Psi^{2}] = \Tr[\Psi] = 1$.
Substituting into \eqref{eq:HS_expand_1},
\eq{
\widetilde{C}_{H_{f}}(\Psi) = \tfrac{M+1}{2} - \Tr[H_{f}\Psi] = \tfrac{M+1}{2} - \brakets{\psi}{H_{f}}{\psi}.\label{eq:HS_expand_final}
}
The first term is a constant determined by the number of marked concepts, so the cost is a strictly decreasing affine function of $\brakets{\psi}{H_{f}}{\psi}$, the projection probability onto the marked subspace.
This probability is maximized, and equals unity, on the solution state $\ket{\psi^{*}} = M^{-1/2}\sum_{x:f(x)=1}\ket{x}$.
Therefore ITE, realized as the Riemannian gradient flow of \eqref{eq:g-cost}, drives the system to the global minimum of $\widetilde{C}_{H_{f}}$, which coincides with the unique pure state lying entirely in the marked subspace.
The convergence rate is controlled by the energy variance along the trajectory, and by Lemma~\ref{lem:ev_ineq} (Appendix~\ref{sec:terminology}) the infidelity $\eps = 1 - F$ to the solution decays exponentially in the imaginary time, with rate set by the spectral gap $\lambda_{1}$ of $H_{f}$.
This is the mechanism that, in the cognitive model, drives an initially uniform mental superposition into a configuration localized on the target concept, whether the process is unconscious or conscious.

This descent also admits a reading through active inference and the free-energy principle.
The oracle $H_{f}$ need not be an externally imposed rule.
It can be viewed as an internalized self-prior, a density model of familiar, homeostatic sensory states acquired through the agent's own developmental experience~\cite{kim2025}, so that the cost-minimizing ITE parallels the minimization of expected free energy, in which the system updates its state to reduce the mismatch between its current configuration and its self-prior.

Psychologically, candidate contents are not selected by an all-or-none switch.
They are reweighted by a cost landscape until one region of representational space dominates, and whether this dominance remains implicit or becomes reportable depends on the measurement coupling discussed next.

\subsection{Unconscious and Conscious Cognition as a Continuum of Measurement Strength}
\label{sec:meas_strength}
We now state the main interpretive result.
In Secs.~\ref{sec:exact_pi}--\ref{sec:hs_expansion} the underlying cognitive dynamics has been identified as the ITE~\eqref{eq:ite_norm} on the oracle $H_{f}$, equivalently the unitary flow~\eqref{eq:Ueff} generated by $\calH_{\text{eff}}$, equivalently the discrete path integral of Main Result~\ref{mr:exact_path_integral}.
This dynamical law applies to both unconscious and conscious cognitive processing.
What differentiates the two regimes is not the form of the dynamics but the \textit{strength of the measurement interaction} between the cognitive system and the neural environment that acts as its probe.

\paragraph*{The indirect measurement model and the strength continuum.}
We adopt the indirect measurement model~\cite{vonneumann1955,ozawa1984,ozawa2023}, in which a measurement of a cognitive observable $\widehat{O}$ on the system $\calH_{S}$ proceeds in two steps.
(i) The system is coupled to a probe $\calH_{P}$ initially in a state $\ket{\phi_{0}}$, through a unitary interaction $U_{\text{int}}(\lambda) = \exp(-\i\lambda \widehat{O}\otimes\widehat{P})$, where $\widehat{P}$ is a probe observable and $\lambda \geq 0$ is the \textit{coupling strength}.
(ii) A projective measurement of $\widehat{P}$ on the probe extracts the outcome.
The coupling strength $\lambda$ interpolates two limits.
\begin{enumerate}[label=(\alph*)]
  \item $\lambda \to 0^{+}$, \textit{weak measurement}: the system is minimally disturbed, only partial statistical information about $\widehat{O}$ is acquired, and coherence between eigenstates of $\widehat{O}$ is preserved over many repetitions.
  \item $\lambda \to \infty$, \textit{strong (projective) measurement}: the system is sharply correlated with the probe, complete information is acquired, and the post-measurement state is an eigenstate of $\widehat{O}$.
\end{enumerate}
Both limits are obtained by varying the same mechanism, the strength of the unitary interaction $U_{\text{int}}(\lambda)$, and there is no ontological discontinuity between them.

\paragraph*{Cognitive interpretation.}
The cognitive system, modeled by $\calH \cong \bbC^{N}$ and evolving under the ITE and path-integral dynamics derived above, is continuously coupled to a neural environment that acts as an internal probe.
The coupling strength $\lambda$ controls the measurement strength in the sense of Refs.~\cite{emori2026cc,emori2026ql}.
When $\lambda$ is small, the cognitive state evolves under ITE with its interference structure intact, and the environment registers only partial information.
This is the unconscious regime, in which a goal-directed optimization continues, still descending the cognitive cost, but no single outcome is reportable and candidates can interfere along the trajectory.
As $\lambda$ increases, the probe registers sharper information about the cognitive state, and in the limit $\lambda \to \infty$ it performs an effective projective measurement, collapsing the optimized superposition into a single reportable thought.
This is the reportable fixation endpoint, of which an insight-like ``Aha'' experience is one phenomenological example.
Intermediate values of $\lambda$ correspond to the preconscious regime, in which candidates are narrowed but not yet fixed.
This coupling strength $\lambda$ is the dynamical realization of measurement strength in Ref.~\cite{emori2026cc}.
Experimentally it corresponds not to a single laboratory control but to a family of constraints that increase readout strength, including report demand, forced choice, time pressure, confidence rating, action selection, and environmental feedback.
\begin{mainresult}[ITE and path-integral structure with a measurement-strength continuum]\label{mr:wick_meas}
Let $\calH_{\text{eff}} = \i[H_{f},\psi_{0}]$ be the effective cognitive Hamiltonian on $\calH = \bbC^{N}$.
Cognitive cost optimization is governed, for all values of the system--probe coupling strength $\lambda \geq 0$, by the same underlying ITE
\eq{
\ket{\Psi(\tau)} = \frac{\e^{\tau H_{f}}\ket{\psi_{0}}}{\norm{\e^{\tau H_{f}}\ket{\psi_{0}}}},
}
equivalently the Wick-rotated unitary $\ket{\psi_{s}} = \e^{-\i s\calH_{\text{eff}}}\ket{\psi_{0}}$, equivalently the exact discrete path integral~\eqref{eq:KN_exact}.
The unconscious-to-conscious transition is the continuous deformation of the coupling strength $\lambda$.
\begin{itemize}[leftmargin=2em]
  \item $\lambda$ small: weak measurement, interference among virtual cognitive trajectories preserved, no unique reportable outcome (unconscious regime);
  \item $\lambda$ moderate: intermediate measurement, candidates progressively rectified (preconscious regime);
  \item $\lambda$ large: strong projective measurement, the optimized superposition collapses to a single eigenstate of the readout basis and a definite conscious report is produced (conscious regime).
\end{itemize}
In all three regimes the cognitive state descends along the Hilbert--Schmidt cost~\eqref{eq:g-cost} and, by Eq.~\eqref{eq:HS_expand_final} and Lemma~\ref{lem:ev_ineq}, converges exponentially fast, in the absence of strong-measurement collapse, to the solution state $\ket{\psi^{*}}$.
\end{mainresult}
The Wick rotation in Sec.~\ref{sec:exact_pi} is therefore a mathematical step, not a physical transition.
It allows the same dynamical content to be written either as a cost-minimizing ITE, the description natural to optimization, or as a unitary path integral, the description natural to interference and to quantum-hardware implementation~\cite{gluza2024,gluza2026,gilyen2019,martyn2021,suzuki2025}.
The physical transition between cognitive regimes is implemented by the measurement-strength continuum $\lambda$, which deforms the coupling between the cognitive system and its neural-environment probe.
In these terms the apparent paradox that the unconscious mind evolves in real time while the conscious mind evolves in imaginary time is dissolved, since both evolve under the same imaginary-time cost descent, and what varies is the strength with which the environment monitors that descent.

\subsection{Weak-Coupling Limit of the Measuring Interaction and GKSL Decoherence}
\label{sec:lindblad}
We make rigorous the relation between the indirect measurement model of Sec.~\ref{sec:meas_strength} and the Gorini--Kossakowski--Sudarshan--Lindblad (GKSL) master equation~\cite{gorini1976,lindblad1976} that underlies the quantum-like decoherence model of cognition introduced by Asano \textit{et al.}~\cite{asano2011}.
For readers in neuroscience, the GKSL equation is a controlled idealization.
It describes the path-integral dynamics when the neural environment is approximated as many fresh, independent, weakly coupled probes, and it is not assumed to be the generic dynamics of the brain.
Two clarifications of terminology are in order.
First, the limit of physical relevance is not a generic weak-measurement limit of the indirect measurement model, but the \textit{weak-coupling limit of the measuring interaction within it}, in which the coupling strength $\lambda$ in the system--probe unitary is sent to zero with the simultaneous time-rescaling $\lambda^{2} = \gamma\,\Delta t$ ($\gamma>0$ fixed, $\Delta t \to 0^{+}$).
The qualifier ``weak measurement'' designates the regime that emerges from this limit, not the limit itself.
Second, the GKSL form is not implied by weak coupling alone, since it additionally requires a Markovian idealization, namely that the neural environment is modeled as a sequence of freshly prepared, statistically independent probes.
Both points are made precise in Theorem~\ref{thm:gksl_limit} below.

\paragraph*{Setup.}
Let $\calH_{S} \cong \bbC^{N}$ denote the cognitive Hilbert space, on which the cognitive system evolves under the effective Hamiltonian $\calH_{\mathrm{eff}} = \i[H_{f},\psi_{0}]$ of Eq.~\eqref{eq:Heff_def}.
We discretize time into intervals of length $\Delta t > 0$.
At each time step the cognitive system interacts with one element of a sequence $\{\calH_{P,n}\}_{n\geq 1}$ of identical probe Hilbert spaces, the $n$-th probe being initially prepared in the same fiducial state $\sigma$ with $\Tr[\sigma]=1$.
The system--probe interaction is the indirect-measurement unitary
\eq{
U_{\lambda} = \e^{-\i\lambda A\otimes B},\label{eq:probe_unitary}
}
with $A = A^{\dagger}$ a cognitive observable on $\calH_{S}$, $B = B^{\dagger}$ a probe observable on $\calH_{P,n}$, and $\lambda \geq 0$ the coupling strength.
Including the free system evolution $e^{-\i\calH_{\mathrm{eff}}\Delta t}$ over the same time step and tracing out the probe, the one-step reduced cognitive dynamics is the completely positive trace-preserving (CPTP) map
\eq{
\Phi_{\Delta t}(\rho) \;\equiv\; \Tr_{P}\!\left[U_{\lambda}\bigl(\e^{-\i\calH_{\mathrm{eff}}\Delta t}\,\rho\, \e^{\i\calH_{\mathrm{eff}}\Delta t}\otimes\sigma\bigr)U_{\lambda}^{\dagger}\right].\label{eq:phi_dt}
}
We adopt the following assumptions throughout this section.
\begin{enumerate}[label=(\roman*),leftmargin=2em]
  \item \textbf{Single use.} Each probe $\calH_{P,n}$ interacts with the cognitive system exactly once.
  \item \textbf{Independent and identical preparation.} The probes are i.i.d., each prepared in the same state $\sigma$.
  \item \textbf{Unconditioned reduction.} Probe outcomes are discarded, and the cognitive dynamics is given by the partial trace in~\eqref{eq:phi_dt}.
  \item \textbf{Zero mean.} $\langle B\rangle_{\sigma} \equiv \Tr[B\sigma] = 0$.
  \item \textbf{Weak-coupling scaling.} The coupling strength is rescaled with the time step as $\lambda^{2} = \gamma\,\Delta t$, with $\gamma > 0$ fixed, and the continuous-time limit $\Delta t \to 0^{+}$ is taken with $t = n\Delta t$ held fixed.
\end{enumerate}
Two remarks on these assumptions are warranted.
Assumption~(iv) is required for the limit~(v) to be well-defined, since otherwise the contribution linear in $\lambda$ to $\Phi_{\Delta t}$ would scale as $\sqrt{\Delta t}$ and diverge upon iteration, and (iv) is harmless because any nonzero $\langle B\rangle_{\sigma}$ can be absorbed into a redefinition of $\calH_{\mathrm{eff}}$.
For convenience and without further loss of generality, we also normalize $\langle B^{2}\rangle_{\sigma} = 1$ by absorbing it into the constant $\gamma$.
Assumptions~(i)--(ii) together encode the Markovian idealization, (iii) selects the unconditioned dynamics, and (v) implements the weak-coupling limit of the measuring interaction.

\paragraph*{Expansion of the one-step CPTP map.}
Expanding $U_{\lambda} = I - \i\lambda(A\otimes B) - \tfrac{\lambda^{2}}{2}(A\otimes B)^{2} + \calO(\lambda^{3})$ and using $[A,[A,\rho]] = \{A^{2},\rho\} - 2A\rho A$, direct computation of the partial trace in~\eqref{eq:phi_dt} gives, to second order in $\lambda$ and first order in $\Delta t$,
\eq{
\Phi_{\Delta t}(\rho) - \rho
&= -\i\Delta t\,[\calH_{\mathrm{eff}},\rho]
- \i\lambda\,\langle B\rangle_{\sigma}\,[A,\rho]\nonumber\\
&\quad + \lambda^{2}\langle B^{2}\rangle_{\sigma}\!\left(A\rho A - \tfrac{1}{2}\{A^{2},\rho\}\right) + \calO(\Delta t^{3/2}).\label{eq:Phi_expand}
}
Assumption~(iv) removes the second term, and (v) substitutes $\lambda^{2} = \gamma\Delta t$.
With the normalization $\langle B^{2}\rangle_{\sigma} = 1$ from~(v), this yields
\eq{
\Phi_{\Delta t}(\rho) = \rho + \Delta t\,\calL(\rho) + \calO(\Delta t^{3/2}),\label{eq:phi_lindblad}
}
with the bounded generator
\eq{
\calL(\rho) \;\equiv\; -\i[\calH_{\mathrm{eff}},\rho] + \gamma\!\left(A\rho A - \tfrac{1}{2}\{A^{2},\rho\}\right).\label{eq:gksl_generator}
}
The operator $\calL$ is of canonical GKSL form with a single Lindblad operator $L = \sqrt{\gamma}\,A$, since $A=A^{\dagger}$ gives $L^{\dagger}L = \gamma A^{2}$ and $L\rho L^{\dagger} = \gamma A\rho A$, so that $L\rho L^{\dagger} - \tfrac{1}{2}\{L^{\dagger}L,\rho\}$ reproduces the dissipator in~\eqref{eq:gksl_generator}.
\begin{theorem}[GKSL limit of repeated indirect measurement]\label{thm:gksl_limit}
Under the setup~\eqref{eq:phi_dt} and Assumptions~(i)--(v),
\eq{
\lim_{\Delta t\to 0^{+}}\Phi_{\Delta t}^{\,\lfloor t/\Delta t\rfloor}(\rho) \;=\; \e^{t\calL}(\rho) \qquad\forall\, t\geq 0,
}
with $\calL$ the GKSL generator~\eqref{eq:gksl_generator}.
Equivalently, the reduced cognitive state $\rho(t)$ obeys the GKSL master equation
\eq{
\frac{\d\rho(t)}{\d t} \;=\; \calL(\rho(t)) \;=\; -\i[\calH_{\mathrm{eff}},\rho(t)] + \gamma\!\left(A\rho(t)A - \tfrac{1}{2}\{A^{2},\rho(t)\}\right).\label{eq:gksl}
}
\end{theorem}
\begin{proof}[Sketch]
The expansion~\eqref{eq:phi_lindblad} expresses $\Phi_{\Delta t}$ as a uniform first-order perturbation of the identity by the bounded linear generator $\calL$ on the finite-dimensional space of density operators.
The semigroup convergence $\Phi_{\Delta t}^{\lfloor t/\Delta t\rfloor} \to e^{t\calL}$ then follows from the Chernoff approximation theorem, equivalently the Trotter--Kato product formula, for strongly continuous one-parameter semigroups of CPTP maps.
The same conclusion is obtained, in the language of repeated quantum interactions and under the present scaling, by the convergence theorem of Attal and Pautrat~\cite{attal2006}; compare the weak-coupling semigroup construction of Davies~\cite{davies1976}.
The derivation of~\eqref{eq:Phi_expand} is a direct second-order calculation.
With $V \equiv A\otimes B$, one has $U_{\lambda}(\rho\otimes\sigma)U_{\lambda}^{\dagger} = \rho\otimes\sigma - \i\lambda[V,\rho\otimes\sigma] - \tfrac{\lambda^{2}}{2}[V,[V,\rho\otimes\sigma]] + \calO(\lambda^{3})$.
Taking the partial trace and using $\Tr[B\sigma]=\langle B\rangle_{\sigma}$, $\Tr[B^{2}\sigma]=\langle B^{2}\rangle_{\sigma}$, $[A,[A,\rho]] = \{A^{2},\rho\}-2A\rho A$, and $e^{-\i\calH_{\mathrm{eff}}\Delta t}\rho e^{\i\calH_{\mathrm{eff}}\Delta t} = \rho - \i\Delta t[\calH_{\mathrm{eff}},\rho]+\calO(\Delta t^{2})$ yields~\eqref{eq:Phi_expand}, and the cross-terms between the two expansions are of order $\lambda\Delta t = \calO(\Delta t^{3/2})$ under (v).
The conditioned, outcome-keeping variant of the same dynamics is the quantum filtering equation of Bouten, van Handel, and James~\cite{bouten2007}, whose outcome average reproduces~\eqref{eq:gksl}.
\end{proof}
\begin{remark}[Multi-channel couplings and Lamb shift]\label{rem:multichannel}
If the indirect-measurement unitary is generalized to $V = \sum_{k} A_{k}\otimes B_{k}$ with finitely many self-adjoint pairs $(A_{k},B_{k})$, and the analogous zero-mean assumption $\langle B_{k}\rangle_{\sigma} = 0$ holds for every $k$, the same argument yields the multi-channel GKSL equation
\eq{
\frac{\d\rho(t)}{\d t} = -\i[\calH_{\mathrm{eff}} + H_{\mathrm{LS}},\rho(t)] + \sum_{j}\!\left(L_{j}\rho(t)L_{j}^{\dagger} - \tfrac{1}{2}\{L_{j}^{\dagger}L_{j},\rho(t)\}\right),\label{eq:gksl_general}
}
where the Lindblad operators $\{L_{j}\}$ are obtained by diagonalizing the positive Hermitian correlation matrix $C_{kl} = \Tr[B_{l}B_{k}\sigma]$ and the Lamb-shift Hamiltonian $H_{\mathrm{LS}}$ arises from the antisymmetric part of $C$.
The structure of the rest of this section is unaffected by this generalization, so we work with the single-channel form~\eqref{eq:gksl}.
\end{remark}
\begin{corollary}[Recovery of the Asano \textit{et al.}~model]\label{cor:asano_limit}
Under Assumptions~(i)--(v), the quantum-like GKSL decoherence model of cognition introduced by Asano \textit{et al.}~\cite{asano2011} is recovered as the Markovian weak-coupling limit (Theorem~\ref{thm:gksl_limit}) of the present cognitive path-integral framework.
The dissipative term in~\eqref{eq:gksl} encodes the cumulative effect of many \textit{unmonitored weak measurements}~\cite{aharonov1988,dressel2014} performed by the neural environment, and the Hamiltonian part is generated by the cognitive Hamiltonian $\calH_{\mathrm{eff}} = \i[H_{f},\psi_{0}]$.
\end{corollary}
The framework identifies three structural levels.
The first is the unitary cognitive path integral of Main Result~\ref{mr:exact_path_integral}, which retains the interference structure of the cognitive state.
The second is the GKSL semigroup of Theorem~\ref{thm:gksl_limit}, the Markovian weak-coupling reduction whose unitary part is the Wick-rotated cognitive ITE generated by $\calH_{\mathrm{eff}}$ and whose dissipative part is the cumulative effect of unmonitored weak measurements, recovering the Asano \textit{et al.}~model.
The third is the strong-coupling limit of Main Result~\ref{mr:wick_meas}, which yields the projective collapse identified with reportable fixation.
The unconscious, decoherent, and conscious stages commonly distinguished in the consciousness literature are, in this picture, not three distinct dynamical laws but three positions on the single continuum $\lambda$, with the GKSL semigroup arising only in the additional Markovian regime captured by Assumptions~(i)--(ii).
This is the sense in which the present paper supplies the formal basis of the measurement-strength framework developed in Refs.~\cite{emori2026cc,emori2026ql}.
\begin{remark}[Non-Markovianity and the genericity of the path integral]\label{rem:nonmarkov}
Markovianity is not implied by the weak-coupling scaling~(v) alone.
Assumptions~(i)--(ii), that the neural environment can be modeled as a sequence of freshly prepared, statistically independent probes, constitute an idealization that is not, in general, physiologically warranted for neuro-cognitive systems, since the neural environment carries history-dependent correlations and there is no microscopic justification for the repeated re-preparation of the same probe state $\sigma$ between successive interactions.
Relaxing Assumption~(ii), and more drastically (i), leads, by tracing out the environmental degrees of freedom, to a non-Markovian reduced dynamics whose generator does not admit GKSL form, corresponding in the path-integral language to a nonlocal Feynman--Vernon-type influence functional.
The GKSL model should therefore be regarded as the Markovian approximation of a more general path-integral theory, and the unitary path integral of Main Result~\ref{mr:exact_path_integral} retains its primacy as the descriptive substrate.
A systematic non-Markovian extension is left for future work.
\end{remark}
A word is in order on the interpretational questions that surround the strong-coupling limit~\cite{vonneumann1955}, namely the location of the Heisenberg cut, the choice of measurement basis by the environment, and the question of wave-function collapse.
We do not claim to resolve any of these.
We claim that the structural correspondence between strong projective measurement and the formation of a conscious report is operative regardless of one's stance on these questions, and that its consequences, in particular the discriminating inequalities that follow from non-commutative cognitive structure~\cite{emori2026cc,emori2026ql}, are independently testable.

The GKSL dynamics~\eqref{eq:gksl} also maintains the cognitive system in a nonequilibrium steady state when continuously driven by weak measurements.
From the standpoint of stochastic thermodynamics, the irreversibility of this open-system evolution is quantified by the entropy production rate, which grows with the strength of the dissipator and hence with the coupling through $\gamma\propto\lambda^{2}/\Delta t$.
This gives a formal reading of recent findings in which large-scale brain dynamics exhibit broken detailed balance and increased time-irreversibility during cognitively demanding tasks and conscious states~\cite{nartallo2026}.

\section{Conclusion}
\label{sec:conclusion}
We have developed a path-integral model of cognition whose three elements are imaginary-time evolution on a projector Hamiltonian as the dynamical law of cognitive cost optimization, the Wick rotation $t = -\i\tau$ as the technique that re-expresses this non-unitary descent as an equivalent unitary evolution admitting an exact discrete path integral on a finite-dimensional Hilbert space, and the strength of the unitary measurement coupling to the neural environment as the physical variable that interpolates between the unconscious and conscious regimes.
Three technical results sharpen this picture.
First (Sec.~\ref{sec:ite_dbf}), ITE on a projector Hamiltonian coincides with the double-bracket flow generated by the commutator $[H_{f},\psi_{0}]$, and is the Riemannian gradient flow of a Hilbert--Schmidt cost function whose unique global minimum is the solution state.
Second (Sec.~\ref{sec:exact_pi}, Main Result~\ref{mr:exact_path_integral}), the Wick-rotated unitary description of this cost optimization admits an exact discrete path-integral representation on $\bbC^{N}$, with the oracle of meaning $H_{f}$ acting as a potential energy and the initial-state diffusion projector $\psi_{0}$ acting as a kinetic energy.
Third (Sec.~\ref{sec:hs_expansion}), the Hilbert--Schmidt cost function expands into an affine, decreasing function of the projection probability $\brakets{\psi}{H_{f}}{\psi}$, so that ITE settles exponentially fast on the unique pure state in the marked subspace.
We then showed (Sec.~\ref{sec:meas_strength}, Main Result~\ref{mr:wick_meas}) that the same ITE and path-integral structure governs cognitive processing in both regimes, and that the transition between them is a continuous deformation of the coupling strength $\lambda$ between the cognitive system and a neural-environment probe.
The Markovian weak-coupling limit of the measuring interaction (Sec.~\ref{sec:lindblad}, Theorem~\ref{thm:gksl_limit} and Corollary~\ref{cor:asano_limit}) recovers the GKSL-type decoherence model of Asano \textit{et al.}~\cite{asano2011} as the open-system reduction of our unitary cognitive path integral under the hypothesis of independent and identically prepared probes, and the strong-coupling limit recovers the moment of insight as the reportable fixation of the optimized superposition.
Without the Markovian assumption the reduction is more general and yields a non-Markovian dynamics (Remark~\ref{rem:nonmarkov}), of which the GKSL semigroup is a special case.

The present paper supplies the mathematical and physical foundations of a broader three-paper measurement-strength program.
Ref.~\cite{emori2026cc} surveys major contemporary theories of consciousness (IIT, GNW, RPT, HOT, PP) and argues that their apparent conflicts can be reorganized as different focal points along a single measurement-strength axis~\cite{tsuchiya2015,cogitate2025}.
Ref.~\cite{emori2026ql} shows why this axis requires a contextual, non-Boolean foundation and how classical Boolean records arise by context-forgetting.
The present paper identifies the dynamical realization of that axis with the coupling-strength parameter $\lambda$ of Main Result~\ref{mr:wick_meas}.
The three papers together propose a single re-coordination of vocabularies, in which contextual logic supplies the space of possible questions, path-integral dynamics supplies the evolution of candidate histories, and consciousness science supplies the empirical domain in which measurement strength is manipulated.

For empirical work, the central prediction is that neural and behavioral signatures vary systematically with readout strength.
Report demand, confidence rating, forced choice, time pressure, delayed report, and closed-loop feedback are candidate experimental handles on $\lambda$, and Ref.~\cite{emori2026cc} develops this point as a theory-comparison program for IIT, GNW, RPT, HOT, and PP.
Several further directions are open.
The exact path integral~\eqref{eq:KN_exact} provides discriminating inequalities for the non-Boolean versus non-Bayesian dichotomy studied in Refs.~\cite{emori2026cc,emori2026ql}, whose experimental evaluation, including in the no-report paradigms surveyed there~\cite{tsuchiya2015}, is a natural next step.
The unitary realization of cognitive ITE on the same Hilbert space, combined with the controlled approximation by product formulae, suggests simulation strategies on near-term quantum hardware using QSVT~\cite{gilyen2019,martyn2021,suzuki2025} and DBQA~\cite{gluza2024,gluza2026}.
The non-Markovian extension of Remark~\ref{rem:nonmarkov}, in which the neural environment carries history-dependent correlations and the reduced cognitive dynamics is governed by a Feynman--Vernon-type influence functional, is a direction for capturing the phenomenology of cognition more faithfully.
The intersubjectivity program initiated by Ozawa's theorem~\cite{ozawa2025,khrennikov2024a,khrennikov2024b}, situated within the triadic niche-construction framework of Iriki and collaborators~\cite{iriki2012,iriki2024}, can be connected to the present model by treating the measurement basis itself as a socially constructed structure, an avenue developed in Refs.~\cite{emori2026cc,emori2026ql}.
Finally, the identification of the moment of insight with the strong-coupling limit of a continuous measurement-strength axis invites empirical comparison with neural correlates of conscious access and with the boundary cases of animals, patients, organoids, and AI, for which the measurement-strength language is clarifying.

\begin{acknowledgments}
We are grateful to Masanao Ozawa for valuable discussions.
We thank the organizers and participants at the Ernst Str{\"{u}}ngmann Forum ``Simplicity behind Absurdity: The Power of Quantum Thinking'' held in Frankfurt in September 2025, where this work was initiated.
This work was supported by JST ASPIRE Grant Number JPMJAP2318, JSPS KAKENHI Grant Numbers 24H02200 and 26H02539, JST SPRING Grant Number JPMJSP2119, and the RIKEN Junior Research Associate Program.
\end{acknowledgments}

\appendix

\section{Common Terminology and Techniques}
\label{sec:terminology}
This appendix collects definitions and lemmas used throughout the main text.
\begin{enumerate}
\item The \textit{Hilbert--Schmidt norm} of an operator $A$ on $\calH$ is $\norm{A}_{\text{HS}} = \sqrt{\Tr[A^{\dagger}A]}$.
It is the Schatten 2-norm, and the inner product $\braket{A,B} = \Tr[A^{\dagger}B]$ makes the operator space a Hilbert space.
\item The \textit{operator norm} of $A$ is $\norm{A}_{\text{op}} = \sup_{\norm{\psi}=1}\norm{A\psi}$.
For a finite-dimensional self-adjoint operator, $\norm{A}_{\text{op}} = \max_{i}|\lambda_{i}|$.
Both norms satisfy the triangle inequality $\norm{A+B}\leq\norm{A}+\norm{B}$.
\item The Taylor expansion of $f$ around $x_0$ with Lagrange remainder reads
\eq{
f(x) = \sum_{k=0}^{n-1}\frac{f^{(k)}(x_0)}{k!}(x-x_0)^{k} + \frac{f^{(n)}(\xi)}{n!}(x-x_0)^{n},
}
for some $\xi$ between $x_0$ and $x$.
\item For pure states $\ket{\phi},\ket{\psi}$, the \textit{fidelity} is $F(\ket{\phi},\ket{\psi}) = |\braket{\phi|\psi}|^{2}$, measuring their overlap.
Convergence of ITE is expressed through the infidelity $\eps = 1 - F$.
\end{enumerate}
\begin{lemma}\label{lem:ev_ineq}
Let $\ket{\psi}$ be a pure state and $H = \sum_{i=0}^{d-1}\lambda_{i}\ketbra{\lambda_i}{\lambda_i}$ a Hamiltonian with increasingly ordered eigenvalues and $\lambda_0 = 0$.
If the ground-state fidelity is $F = 1 - \eps$, then
\eq{
E_{k} \geq \lambda_{1}\eps,\qquad V_{k} \leq \lambda_{d-1}^{2}\eps,
}
where $E_{k} = \brakets{\psi}{H}{\psi}$ and $V_{k} = \brakets{\psi}{H^{2}}{\psi} - E_{k}^{2}$.
\end{lemma}
\begin{proofof}{Lemma~\ref{lem:ev_ineq}}
Set $p_i = |\braket{\lambda_i|\psi}|^{2}$, so $p_0 = 1-\eps$ and $\sum_{i\geq1}p_i = \eps$.
Then
\eq{
E_{k} = \sum_{i\geq 1}\lambda_i p_i \geq \lambda_1\sum_{i\geq 1}p_i = \lambda_1\eps.
}
For the variance,
\eq{
V_{k} \leq \sum_{i\geq 1}\lambda_i^{2}p_i \leq \lambda_{d-1}^{2}\sum_{i\geq 1}p_i = \lambda_{d-1}^{2}\eps.
}
\end{proofof}
We also record the derivation of Eq.~\eqref{eq:rg}.
The tangent space at $X = UAU^{\dagger} \in \calM(A)$ is $T_{X}\calM(A) = \{[X,\xi]\mid \xi^{\dagger} = -\xi\}$, and the Riemannian metric $\braket{[X,\Omega_1],[X,\Omega_2]} = \Tr[\Omega_1^{\dagger}\Omega_2]$ is induced by the HS structure.
Differentiating $C_{B}(X) = -\tfrac{1}{2}\norm{X-B}^{2}_{\text{HS}}$ in the direction $[X,\Omega]$ and equating with the inner product $\braket{[X,\xi],[X,\Omega]} = \Tr[\xi^{\dagger}\Omega]$ yields $\xi = [X,B]$, whence \eqref{eq:rg} follows.

\section{Grover's Algorithm}
\label{sec:ga}
We review Grover's algorithm~\cite{grover1996} for unstructured database search, treated as a worked example of the framework of Sec.~\ref{sec:pim}.
Given $\calX = \{0,1,\dots,N-1\}$ with $N = 2^{n}$ and an oracle $f:\calX\to\{0,1\}$ marking $M$ items, prepare
\eq{
\ket{\psi_0} = \mathsf{H}^{\otimes n}\ket{0^{\otimes n}} = \frac{1}{\sqrt{N}}\sum_{x}\ket{x},
}
and iterate the Grover operator
\eq{
G_{k}(\alpha_k,\beta_k) = -D(\alpha_k)U_{f}(\beta_k),
\quad D(\alpha) = I - (1-\e^{\i\alpha})\psi_0,
\quad U_{f}(\beta) = I - (1-\e^{\i\beta})H_{f},
}
to approximate the target state $\ket{\psi^{*}} = M^{-1/2}\sum_{x:f(x)=1}\ket{x}$.
The original choice $\alpha_k = \beta_k = \pi$ achieves the asymptotically optimal $\calO(\sqrt{N/M})$ query complexity~\cite{bennett1997,beals1998,zalka1999}.
Refinements include the $\pi/3$-algorithm~\cite{grover2005}, which addresses the soufflé problem~\cite{brassard1997} at the cost of the quadratic speedup, and fixed-point search~\cite{yoder2014,li2026}, which restores optimality through recursive quasi-Chebyshev phase schedules
\eq{
\alpha_{k} = \beta_{\calN-k+1} = -\cot^{-1}\!\left[\tan\!\left(\frac{2\pi k}{\calN}\right)\sqrt{1 - \frac{1}{\gamma^{2}}}\right],
}
with $\gamma = T_{1/\calN}(1/\delta)$ for a desired fidelity $F\geq 1 - \delta^{2}$.
Deterministic variants achieving zero-error final states have also been developed~\cite{long2001,roy2022}.
The geometric and information-theoretic structure of Grover dynamics has been examined~\cite{miyake2001,cafaro2012}, and the framework has been generalized to amplitude amplification~\cite{brassard2002} and unified within QSVT~\cite{gilyen2019,martyn2021,suzuki2025}.

\section{Grover's Algorithm as an Approximation of ITE}
\label{sec:ga=ite}

\subsection{ITE Solves Unstructured Search}
\label{sec:ite_solves}
\begin{lemma}[ITE solves unstructured search]\label{lem:ite_solves}
Given $H_{f}$ and $\ket{\psi_0}$,
\eq{
\lim_{\tau\to\infty}\frac{\e^{\tau H_{f}}\ket{\psi_0}}{\norm{\e^{\tau H_{f}}\ket{\psi_0}}} = \ket{\psi^{*}}.
}
\end{lemma}
The positive sign reflects that unstructured search seeks the largest-eigenvalue subspace of $H_{f}$, and via $H \to -H_{f}$ this is the standard ITE convergence statement~\cite{gellmann1951}.

\subsection{Grover's Algorithm as a DBQA}
\label{sec:dbqa}
The double-bracket quantum algorithm~\cite{gluza2024,gluza2026} implements DBF on a quantum computer.
For projector Hamiltonians, the first-order approximation $\ket{\psi_s} = \e^{s[H,\psi_0]}\ket{\psi_0}$ is exact.
\begin{lemma}[Equivalence of ITE and commutator flow for projectors]\label{lem:ite_commutator}
For the projector $H_{f}$ and any ITE time $\tau$, there exists $s_\tau$ such that
\eq{
\frac{\e^{\tau H_{f}}\ket{\psi_0}}{\norm{\e^{\tau H_{f}}\ket{\psi_0}}} = \e^{s_\tau[H_{f},\psi_0]}\ket{\psi_0} = \ket{\psi_{s_\tau}}.\label{eq:dbr}
}
\end{lemma}
The exponential of the commutator is then approximated by product formulae~\cite{dawson2006,chen2022,childs2013}.
The group-commutator identity reads
\eq{
\e^{s[H_{f},\psi_0]} = \e^{\i\sqrt{s}\psi_0}\e^{\i\sqrt{s}H_{f}}\e^{-\i\sqrt{s}\psi_0}\e^{-\i\sqrt{s}H_{f}} + \calO(s^{3/2}),
}
and higher orders give
\eq{
\e^{s[H_{f},\psi_0]}\ket{\psi_0}
\approx \e^{\i t_{2\calN}\psi_0}\cdots\e^{\i t_3 H_{f}}\e^{\i t_2 \psi_0}\e^{\i t_1 H_{f}}\ket{\psi_0}
= (-1)^{\calN}\prod_{k=1}^{\calN}G_{k}(t_{2k},t_{2k-1})\ket{\psi_0},
}
with $t_{k} = c_{k}\sqrt{s}$.
Grover iterations are therefore product-formula approximations of the ITE commutator exponential.
Exact ITE is monotonic and free of overshoot, and the soufflé problem of fixed-angle Grover follows from the accumulation of product-formula errors away from the exact ITE trajectory.

\section{Geodesics and ITE Dynamics for Unstructured Search}
\label{sec:geo}
Using the ITE formulation, the geometry of unstructured search is transparent.
Define the orthogonal complement state
\eq{
\ket{\psi_{0}^{\perp}} = \frac{[H_{f},\psi_0]}{\sqrt{V_0}}\ket{\psi_0} = \frac{H_{f} - E_0 I}{\sqrt{E_0(1-E_0)}}\ket{\psi_0},
}
with $E_0 = \brakets{\psi_0}{H_{f}}{\psi_0} = M/N$ and $V_0 = E_0(1-E_0)$.
\begin{theorem}[ITE traces a geodesic]\label{the:ite_traces_geod}
For projector $H_{f}$ and uniform $\ket{\psi_0}$, the ITE state \eqref{eq:dbr} satisfies
\eq{
\ket{\psi_s} = \cos(s\sqrt{V_0})\ket{\psi_0} + \sin(s\sqrt{V_0})\ket{\psi_{0}^{\perp}},\label{eq:psi_s}
}
which traces the geodesic on $\bbC P^{N-1}$ from $\ket{\psi_0}$ to $\ket{\psi^{*}}$.
The optimal duration is
\eq{
s^{*} = \frac{\arccos(\sqrt{E_0})}{\sqrt{V_0}}.\label{eq:opt_s}
}
\end{theorem}
ITE for a projector Hamiltonian coincides with a geodesic, whereas ITE for a general Hamiltonian is steepest descent of the HS cost~\cite{gluza2026,mcmahon2025} and depends on the cost, so a geodesic is defined independently of any cost.
This coincidence is a special property of the projector case and underlies the optimality results recovered in Appendix~\ref{sec:ga=ite}.

\bibliography{ref}

\end{document}